\documentclass[letter,11pt]{article}%
\usepackage{amsfonts}
\usepackage{amsmath}
\usepackage{amssymb}
\usepackage{graphicx}%
\providecommand{\U}[1]{\protect\rule{.1in}{.1in}}

\providecommand{\U}[1]{\protect\rule{.1in}{.1in}}
\providecommand{\U}[1]{\protect\rule{.1in}{.1in}}
\newtheorem{theorem}{Theorem}
\newtheorem{lemma}[theorem]{Lemma}
\newtheorem{proposition}[theorem]{Proposition}

\newenvironment{proof}[1][Proof]{\noindent\textbf{#1.} }{\ \rule{0.5em}{0.5em}}
\newcommand{\beq}{\begin{equation}}
\newcommand{\eeq}{\end{equation}}

\newenvironment{keywords}
{\begin{trivlist}\item[]{\bfseries\sffamily Keywords:}\ }
{\end{trivlist}}
\begin{document}

\title{Approximation Schemes for Minimizing the Maximum Lateness on a Single Machine
with Release Times under Non-Availability or Deadline Constraints}
\author{Imed Kacem$^{\ast}$, Hans Kellerer$^{+}$\\$^{\ast}$ Laboratoire de Conception, Optimisation et Mod\'{e}lisation des
Syst\`{e}mes, \\LCOMS EA 7306, Universit\'{e} de Lorraine, Metz, France\\\texttt{imed.kacem@univ-\texttt{l}orraine.fr}\\$^{+}$ Institut f\"ur Statistik und Operations Research,\\ISOR, Universit\"{a}t Graz, Graz, Austria\\\texttt{hans.kellerer@uni-graz.at}}
\maketitle

\begin{abstract}
In this paper, we consider four single-machine scheduling problems with
release times, with the aim of minimizing the maximum lateness. In the first
problem we have a common deadline for all the jobs. The second problem looks
for the Pareto frontier with respect to the two objective functions maximum
lateness and makespan. The third problem is associated with a non-availability
constraint. In the fourth one, the non-availibility interval is related to the
operator who is organizing the execution of jobs on the machine (no job can
start, and neither can complete during the operator non-availability period).
For each of the four problems, we establish the existence of a polynomial time
approximation scheme (PTAS).

\end{abstract}


\section{Introduction}

The problem we consider is the one-machine scheduling problem with release
dates and delivery times. The objective is to minimize the maximum lateness.
Formally, the problem is defined in the following way. We have to schedule a
set $J=\{1,2,\ldots,n\}$ of $n$ jobs on a single machine. Each job $j\in J$
has a processing time $p_{j}$, a release time (or head) $r_{j}$ and a delivery
time (or tail) $q_{j}$. The machine can only perform one job at a given time.
Preemption is not allowed. The problem is to find a sequence of jobs, with the
objective of minimizing the maximum lateness $L_{\max}=\max_{1\leq j\leq
n}\left\{  C_{j}+q_{j}\right\}  $ where $C_{j}$ is the completion time of job
$j$. We also define $s_{j}$ as the starting time of job $j$, i.e.,
$C_{j}=s_{j}+p_{j}$, then $P=\sum_{j=1}^{n}p_{j}$ as the total processing time
and $p(H)=\sum_{j\in H}p_{j}$ as the total processing time for a subset of
jobs $H\subset J$. Four scenarios are considered:

\begin{itemize}
\item Scenario 1: every job should be completed before a deadline $d$, i.e.,
$C_{\max}=\max_{1\leq j\leq n}\{C_{j}\}\leq d$. This scenario is denoted as
$1|r_{j},C_{\max}\leq d|L_{\max}$. For a simpler notation, this variant is
also abbreviated as $\Pi_{1}$.

\item Scenario 2: we are given the one-machine bicriteria scheduling problem
with the two objectives $L_{\max}$ and $C_{\max}$, denoted by $1|r_{j}%
|L_{\max}, C_{\max}$. We will speak shortly of problem $\Pi_{2}$.

\item Scenario 3: the machine is not available during a given time interval
$]T_{1},T_{2}[$. This scenario is denoted by $1,h_{1}|r_{j}|L_{\max}$. We will
speak shortly of problem $\Pi_{3}$. Here, $]T_{1},T_{2}[$ is a machine
non-availability (MNA) interval.

\item Scenario 4: in this case, the non-availibility interval $]T_{1},T_{2}[$
is related to the operator who is organizing the execution of jobs on the
machine. An operator non-availability (ONA) period is an open time interval in
which no job can start, and neither can complete. Using an extended notation
as in \cite{Seifa}, this scenario is denoted by $1,ONA|r_{j}|L_{\max}$. We
will speak shortly of proble~$\Pi_{4}$.
\end{itemize}

Note that the main difference between machine non-availability and operator
non-availability consists in the fact that a job can be processed but can
neither start nor finish during the ONA period. However, the MNA interval is a
completely forbidden period.

All four presented scenarios are generalizations of the well-known problem
$1|r_{j},q_{j}|L_{max}$ which has been widely studied in the literature. For
the sake of simplicity, this problem will be denoted by $\Pi_{0}$. According
to Lenstra et al~\cite{Lenstra} problem $1|r_{j},q_{j}|L_{max}$ is NP-hard in
the strong sense. Therefore, we are interested in the design of efficient
approximation algorithms for our problems.

For self-consistency, we recall some necessary definitions related to the
approximation area. A $\rho$-\textit{approximation algorithm} for a problem of
minimizing an objective function $\mathcal{\varphi}$ is an algorithm such that
for every instance $\pi$ of the problem it gives a solution $S_{\pi}$
verifying $\mathcal{\varphi}\left(  S_{\pi}\right)  /\mathcal{\varphi}\left(
OPT_{\pi}\right)  \leq\rho$ where $OPT_{\pi}$ is an optimal solution of $\pi$.
The value $\rho$ is also called the \textit{worst-case bound} of the above algorithm.

A class of $\left(  1+\varepsilon\right)  $-approximation algorithms is called
a \textit{PTAS} (\textit{Polynomial Time Approximation Scheme}), if its
running time is polynomial with respect to the length of the problem input for
every $\varepsilon>0$. A class of $\left(  1+\varepsilon\right)
$-approximation algorithms is called an \textit{FPTAS} (\textit{Fully
Polynomial Time Approximation Scheme}), if its running time is polynomial with
respect to both the length of the problem input and $1/\varepsilon$ for every
$\varepsilon>0$.

Given the aim of this paper, we summarize briefly related results on
scheduling problems with non-availability constraints but no release times. A
huge number of papers is devoted to such problems in the literature. This has
been motivated by practical and real industrial problems, like maintenance
problems or the occurrence of breakdown periods. For a survey we refer to the
article by~\cite{Lee3}. However, to the best of our knowledge, there is only a
limited number of references related to the design of approximation algorithms
in the case of maximum lateness minimization. Yuan et al~\cite{YUAN} developed
a PTAS for the problem without release dates. The paper by Kacem et
al~\cite{Seifa} contains an FPTAS for $1,h_{1}||L_{\max}$ and $1,ONA||L_{\max
}$. Kacem and Kellerer \cite{KKTCS} consider different close semi-online
variants ($1,h_{1}||C_{\max}$, $1,h_{1}|r_{j}|C_{\max}$ and $1,h_{1}||L_{\max
}$) and propose approximation algorithms with effective competitive ratios.

Numerous works address problems with release times but without unavailability
constraints. Most of the exact algorithms are based on enumeration techniques.
See for instance the papers by Dessouky and Margenthaler~\cite{dessouky},
Carlier et al~\cite{CMHG}, Larson et al~\cite{Larson} and Grabowski et
al~\cite{Grabowski}.

Various approximation algorithms were also proposed. Most of these algorithms
are based on variations of the extended Jackson's rule, also called Schrage's
algorithm. Schrage's algorithm consists in scheduling ready jobs on the
machine by giving priority to the one having the greatest tail. It is
well-known that the Schrage sequence yields a worst-case performance ratio of
$2$. This was first observed by Kise et al \cite{Kise}. Potts~\cite{Potts}
improves this result by running Schrage's algorithm at most $n$ times to
slightly varied instances. The algorithm of Potts has a worst-case performance
ratio of $\frac{3}{2}$ and it runs in $O(n^{2}\log n)$ time. Hall and
Shmoys~\cite{Hall} showed that a modification of the algorithm of Potts has
the same worst-case performance ratio under precedence constraints. Nowicki
and Smutnicki \cite{Nowicki} proposed a faster $\frac{3}{2}$-approximation
algorithm with $O(n\log n)$ running time. By performing the algorithm of Potts
for the original and the inverse problem and taking the best solution Hall and
Shmoys \cite{Hall} established the existence of a $\frac{4}{3}$-approximation.
They also proposed two polynomial time approximation schemes. A more effective
PTAS has been proposed by Mastrolilli \cite{Mastrolilli} for the
single-machine and parallel-machine cases. For more details the reader is
invited to consult the survey by Kellerer in \cite{kellererML}.

There are only a few papers which treat the problem with both release times
and non-availability intervals. Leon and Wu~\cite{Leon} present a
branch-and-bound algorithm for the the problem with release times and several
machine non-availability intervals. A 2-approximation has been established by
Kacem and Haouari~\cite{KH} for $1,h_{1}|r_{j}|L_{\max}$.

Rapine et al. \cite{Rapine} have described some applications of operator
non-availability problems in the planning of a chemical experiments. Brauner
et al. \cite{Brauner} have studied single-machine problems under the operator
non-availability constraint. Finally, we refer to the book by T'kindt and
Billaut~\cite{tkindt} for an introduction into multicriteria scheduling.

For each of the four problems $\Pi_{1}$, $\Pi_{2}$, $\Pi_{3}$ and $\Pi_{4}$ we
will present a polynomial time approximation scheme. Notice, that the PTAS's
in~ \cite{Hall} and~ \cite{Mastrolilli} for $1|r_{j},q_{j}|L_{max}$ use
modified instances with only a constant number of release dates. This is not
possible in the presence of deadlines or non-availability periods since the
feasibility of a solution would not be guaranteed.

The paper is organized as follows. Section~2 repeats some notations and
results for Schrage's sequence. Section~3 contains the PTAS for problem
$\Pi_{1}$, Section~4 is devoted to the bicriteria scheduling problem $\Pi_{2}%
$, Sections~5 and~6 deal with problems $\Pi_{3}$ and $\Pi_{4}$, respectively.
Finally, Section~7 concludes the paper.

\section{Schrage's Algorithm}

\label{sec:prel}

For self-consistency we repeat in this section some notations and results for
the Schrage's sequence, initially designed for problem $\Pi_{0}$ and later
exploited for other extensions (see for instance Kacem and Kellerer
\cite{KKDam}). Recall that Schrage's algorithm schedules the job with the
greatest tail from the available jobs at each step. At the completion of such
a job, the subset of the available jobs is updated and a new job is selected.
The procedure is repeated until all jobs are scheduled. For a given instance
$I$ the sequence obtained by Schrage shall be denoted by $\sigma_{Sc}(I)$.
Assume that jobs are indexed such that $\sigma_{Sc}(I)=(1,2,...,n)$. The job
$c$ which attains the maximum lateness in Schrage's schedule, is called the
\emph{critical job\/}. Then the maximum lateness of $\sigma_{Sc}(I)$ can be
given as follows:%

\[
L_{max}(\sigma_{Sc}(I))=\min_{j\in\Lambda}\{r_{j}\}+\sum_{j\in\Lambda}{p_{j}%
}+q_{c}=r_{a}+\sum_{j=1}^{c}{p_{j}}+q_{c},
\]
where job $a$ is the first job so that there is no idle time between the
processing of jobs $a$ and $c$, i.e., either there is idle time before $a$ or
$a$ is the first job to be scheduled. The sequence of jobs $a,a+1,\ldots,c$ is
called the \emph{critical path\/} in the Schrage schedule (or the
\emph{critical bloc} $\Lambda$). It is obvious that all jobs $j$ in the
critical path have release dates $r_{j}\geq r_{a}$.

Let $L^{*}_{max}(I)$ denote the optimal maximum lateness for a given instance
$I$. We will write briefly $L^{*}_{max}$ if it is clear from the context. For
every subset of jobs $F \subset J$ the following useful lower bound is valid.%

\begin{equation}
L^{*}_{max} \geq\min_{j \in F}\{r_{j}\} + \sum_{j \in F}{p_{j}} + \min_{j \in
F}\{q_{j}\} \label{eq66}%
\end{equation}

If $c$ has the smallest tail in $\Lambda=\{a,a+1,\ldots,c\}$, sequence
$\sigma_{Sc}(I)$ is optimal by inequality~(\ref{eq66}). Otherwise, there
exists an \emph{interference job\/} $b\in\Lambda$ such that%

\begin{equation}
q_{b} < q_{c} \ \mbox{and}\ q_{j}\geq q_{c}\ \mbox{for all} \ j \in\{b + 1, b
+ 2,\ldots, c-1\}. \label{eq4}%
\end{equation}

Let $\Lambda_{b}:=\{b+1,b+2,\ldots,c\}$ be the jobs in $\Lambda$ processed
after the interference job $b$. Clearly, $q_{j}\geq q_{c}>q_{b}$ and
$r_{j}>s_{b}$ hold for all $j\in\Lambda_{b}$, where $s_{b}$ denotes the
starting time of the interference job $b$. Inequality (\ref{eq66}) applied to
$\Lambda_{b}$ gives with (\ref{eq4}) the lower bound
\begin{equation}
L_{max}^{\ast}\geq\min_{j\in\Lambda_{b}}r_{j}+p(\Lambda_{b})+\min_{j\in
\Lambda_{b}}q_{j}>s_{b}+p(\Lambda_{b})+q_{c}. \label{eq4a}%
\end{equation}
where $p(\Lambda_{b})$ is the sum of processing times of jobs in $\Lambda_{b}$.

Since there is no idle time during the execution of the jobs of the critical
sequence, the maximum lateness of Schrage's schedule is $L_{\max}(\sigma
_{Sc}(I))=s_{b}+p_{b}+p(\Lambda_{b})+q_{c}$. Subtracting this equation from
(\ref{eq4a}) we get the following upper bound for the absolute error in terms
of the processing time $p_{b}$ of the interference job $b$
\begin{equation}
L_{max}(\sigma_{Sc}(I))-L_{max}^{\ast}<p_{b}. \label{eq5}%
\end{equation}

Notice that using $L^{*}_{max}\geq p_{b}$ and applying inequality (\ref{eq5})
shows that Schrage yields a relative performance guarantee of 2.

It is well-known that Schrage's sequence is optimal with respect to the
makespan, i.e.,%

\begin{equation}
C_{\max}(\sigma_{Sc}(I))=C_{max}^{\ast}. \label{eq5a}%
\end{equation}

\section{PTAS for the first scenario: $1|r_{j},C_{j}\leq d|L_{\max}$}

In this section, we consider problem $\Pi_{1}$ with a common deadline $d$ for
the jobs. By~(\ref{eq5a}) the Schrage sequence shows whether a feasible
solution for $\Pi_{1}$ exists or not. If $C_{\max}(\sigma_{Sc}(I))>d$, then
the problem has no feasible solution. Hence, we will assume in the following that%

\begin{equation}
C_{\max}(\sigma_{Sc}(I))\leq d \label{eq00}%
\end{equation}

Let $\varepsilon>0$. A job $j$ is called \emph{large\/} if $p_{j}%
\geq\varepsilon L_{\max}(\sigma_{Sc}(I))/2$, otherwise it is called
\emph{small}. Let $L$ be the subset of large jobs. Since $L_{\max}(\sigma
_{Sc}(I))\leq2L_{\max}^{\ast}$ (see~\cite{Kise}), it can be observed that
$|L|\leq2/\varepsilon$. Let $k=|L|$. We assume that jobs are indexed such that
$L=\{1,2,\ldots,k\}$.

Our PTAS is based on the construction of a set of modified instances starting
from the original instance $I$. Set $R=\{r_{1},\ldots,r_{n}\}$ and
$Q=\{q_{1},\ldots,q_{n}\}$. Let us define the following sets of heads and tails:%

\begin{align*}
R\left(  i\right)   &  =\{r_{j}\in R|r_{j}\geq r_{i}\},\quad i=1,2,\ldots,k,\\
Q(i)  &  =\{q_{j}\in Q|q_{j}\geq q_{i}\}, \quad i=1,2,\ldots,k.
\end{align*}

Now, we define the set of all the combinations of couples $(r,q)\in R\left(
i\right)  \times Q(i)$ for every $i=1,2,\ldots,k$:%

\[
W=\{(\widetilde{r}_{1},\widetilde{q}_{1},\widetilde{r}_{2},\widetilde{q}%
_{2},\ldots,\widetilde{r}_{k},\widetilde{q}_{k})|\widetilde{r}_{i}\in R\left(
i\right)  ,\widetilde{q}_{i}\in Q\left(  i\right)  , i=1,2,\dots, k\}\text{ }%
\]

Clearly, the cardinality of $W$ is bounded as follows: $|W|\leq n^{2k}%
=O(n^{4/\varepsilon})$. Let $w\in W$ with $w=(\widetilde{r}_{1},\widetilde{q}%
_{1},\widetilde{r}_{2},\widetilde{q}_{2},\ldots,\widetilde{r}_{k}%
,\widetilde{q}_{k})$. Instance $I_{w}$ is a slight modification of instance
$I$ and is defined as follows. $I_{w}$ consists of large and small jobs. The
$k$ large jobs have modified release times $\widetilde{r}_{1},\widetilde{r}%
_{2},\ldots,\widetilde{r}_{k}$ and delivery times $\widetilde{q}%
_{1},\widetilde{q}_{2},\ldots,\widetilde{q}_{k}$. All processing times and the
$n-k$ small jobs remain unchanged. Let $\mathcal{I}$ be the set of all
possible instances $I_{w}$, i.e., $\mathcal{I}=\{I_{w}|w\in W\}$. It is clear
that $|\mathcal{I}|$ is in $O(n^{4/\varepsilon})$. By the modification, the
processing times of instances in $\mathcal{I}$ are not changed and release
times and delivery times are not decreased.

Let $\widetilde{\mathcal{I}}\subseteq\mathcal{I}$. An instance $I^{\prime}%
\in\widetilde{\mathcal{I}}$ is called \emph{maximal\/} if there is no other
instance $I^{\prime\prime}\in\widetilde{\mathcal{I}}$ such that for every
$i=1,2,\ldots,k$, we have $r_{i}^{\prime}\leq r_{i}^{\prime\prime}$,
$q_{i}^{\prime}\leq q_{i}^{\prime\prime}$ and at least one inequality is
strict. Here, $r_{i}^{\prime}, q_{i}^{\prime}$ denote the heads and tails in
$I^{\prime}$ and $r_{i}^{\prime\prime}, q_{i}^{\prime\prime}$ the heads and
tails in $I^{\prime\prime}$, respectively.

Now, we can introduce our procedure PTAS1 which depends on the instance $I$,
the deadline $d$ and the accuracy $\varepsilon$. \bigskip\ 

\noindent\textbf{Algorithm $PTAS1(I,d,\varepsilon$}\textbf{)}

\begin{description}
\item[\textsc{{Input:}}] An instance $I$ of $n$ jobs, a deadline $d$ and an
accuracy $\varepsilon$.

\item[\textsc{{Output:}}] A sequence $\sigma(I)$ with $L_{\max}(\sigma
(I))\leq(1+\varepsilon)L_{\max}^{\ast}$ and $C_{\max}(\sigma(I))\leq d$.
\end{description}

\begin{enumerate}
\item Run Schrage's algorithm for all instances in $\mathcal{I}$.

\item Select the best feasible solution, i.e., the best solution with
$C_{\max}\leq d$. Apply the corresponding sequence to the original instance
$I$.
\end{enumerate}

\medskip\noindent Let $\sigma^{\ast}\left(  I\right)  $ be a sequence for
instance $I$ which is optimal for $L_{\max}$ under the constraint $C_{\max
}(\sigma^{\ast}(I))\leq d$. Let $\widetilde{I}\in\mathcal{I}$ be an instance
which is \emph{compatible\/} with $\sigma^{\ast}\left(  I\right)  $, i.e.,%

\begin{equation}
\label{comp0}C_{\max}(\sigma^{\ast}(\widetilde{I}))\leq d,
\end{equation}

\begin{equation}
\label{comp1}\widetilde{r}_{j}\leq s_{j}(\sigma^{\ast}\left(  I\right)  )
\text{ for }j=1,2,\ldots,k,
\end{equation}

\begin{equation}
\label{comp2}L_{\max}( \sigma^{\ast}(\widetilde{I})) = L_{\max}(\sigma^{\ast
}(I)).
\end{equation}

The set of all instances $\widetilde{I}\in\mathcal{I}$ which are compatible
with $\sigma^{\ast}$ is denoted as $\mathcal{I}_{\sigma^{\ast}}$. By
(\ref{eq00}) instance $I$ fulfills the conditions (\ref{comp0}) to
(\ref{comp2}). Thus, $\mathcal{I}_{\sigma^{\ast}}$ is nonempty.

\begin{theorem}
\label{theo1} Algorithm PTAS1 yields a $(1+\varepsilon)$-approximation for
$\Pi_{1}$ and has polynomial running time for fixed $\varepsilon$.
\end{theorem}

\begin{proof}
Let $I_{\max}$ be a maximal instance in $\mathcal{I}_{\sigma^{\ast}}$.
Applying Schrage's algorithm to $I_{\max}$ gives the sequence $\sigma
_{Sc}(I_{\max})$. We show%

\begin{equation}
C_{\max}(\sigma_{Sc}(I_{\max}))\leq d \label{proof1}%
\end{equation}
and
\begin{equation}
L_{\max}(\sigma_{Sc}(I_{\max}))\leq(1+\varepsilon)L_{\max}(\sigma^{\ast}(I)).
\label{proof2}%
\end{equation}

From (\ref{eq5a}) and (\ref{comp0}) we conclude that $C_{\max}(\sigma
_{Sc}(I_{\max}))\leq C_{\max}(\sigma^{\ast}(I_{\max}))\leq d$ and
(\ref{proof1}) follows.

Several cases are distinguished:

- Case 1: $\sigma_{Sc}(I_{\max})$ is optimal, i.e., $L_{\max}(\sigma
_{Sc}(I_{\max}))=L_{\max}(\sigma^{\ast}(I))$ and by (\ref{comp2}) we get
(\ref{proof2}).

- Case 2: $\sigma_{Sc}(I_{\max})$ is not optimal. Thus, there is an
interference job $b$. Recall that $\Lambda_{b}=\{b,b+1,\ldots,c\}$ is the set
of jobs processed after the interference job until the critical job $c$. By
(\ref{eq5}) $L_{\max}(\sigma_{Sc}(I_{\max}))-L_{\max}(\sigma^{\ast}(I_{\max
}))<p_{b}$. If $b$ is not large, then (\ref{proof2}) follows immediately with
(\ref{eq5}). Otherwise, two subcases occur.

\ \ Subcase 2.a: $\sigma_{Sc}(I_{\max})$ is not optimal and at least one job
of $\Lambda_{b}$ is processed before $b$ in the optimal solution. Thus, $b$
cannot start before $r_{\min}=\min_{i\in\Lambda_{b}}\{r_{i}\}>s_{b}$.
Consequently, $r_{b}$ can be increased to $r_{\min}$ and the new instance
fulfills the conditions (\ref{comp0}) to (\ref{comp2}), which contradicts the
maximality of $I_{\max}$.

\ \ Subcase 2.b: $\sigma_{Sc}(I_{\max})$ is not optimal and all the jobs of
$\Lambda_{b}$ are processed after $b$ in the optimal solution. Thus, $q_{b}$
can be increased to $q_{c}$ and the new instance fulfills the conditions
(\ref{comp0}) to (\ref{comp2}), which again contradicts the maximality of
$I_{\max}$.

Thus, we have found a sequence which is a $(1+\varepsilon)$-approximation for
$\Pi_{1}$. Since $\mathcal{I}$ has at most $O(n^{4/\varepsilon})$ elements,
algorithm PTAS1 runs in polynomial time.
\end{proof}

\section{PTAS for the second scenario: $1|r_{j}|L_{\max}, C_{\max}$}

Computing a set of solutions which covers all possible trade-offs between
different objectives can be understood in different ways. We will define this
(as most commonly done) as searching for ``efficient'' solutions. Given an
instance $I$ a sequence $\sigma(I)$ \emph{dominates\/} another sequence
$\sigma^{\prime}(I)$ if
\begin{equation}
\label{var:dom}L_{\max}(\sigma(I))\leq L_{\max}(\sigma^{\prime}%
(I))\mbox{ and } C_{\max}(\sigma(I))\leq C_{\max}(\sigma^{\prime}(I))
\end{equation}
and at least one of the inequalities~(\ref{var:dom}) is strict. The sequence
$\sigma(I)$ is called \emph{efficient} or \emph{Pareto optimal} if there is no
other sequence which dominates $\sigma(I)$. The set $\mathcal{P}$ of efficient
sequence for $I$ is called \textit{Pareto frontier}.

We have also to define what we mean by an approximation algorithm with
relative performance guarantee for our bicriteria scheduling problem. A
sequence $\sigma(I)$ is called a \emph{$(1+\varepsilon)$-approximation
algorithm\/} of a sequence $\sigma^{\prime}(I)$ if
\begin{equation}
\label{pareto1}L_{\max}(\sigma(I))\leq(1+\varepsilon) L_{\max}(\sigma^{\prime
}(I))
\end{equation}
and
\begin{equation}
\label{pareto2}C_{\max}(\sigma(I))\leq(1+\varepsilon) C_{\max}(\sigma^{\prime
}(I))
\end{equation}
hold.

A set $\mathcal{F_{\varepsilon}}$ of schedules for $I$ is called a
\emph{$(1+\varepsilon)$-approximation of the Pareto frontier\/} if, for every
sequence $\sigma^{\prime}(I)\in\mathcal{P}$, the set $\mathcal{F_{\varepsilon
}}$~contains at least one sequence $\sigma(I)$ that is a $(1+\varepsilon
)$-approximation of $\sigma^{\prime}(I)$.

A \emph{PTAS for the Pareto frontier\/} is an algorithm which outputs for
every $\varepsilon>0$, a $(1+\varepsilon)$-approximation of the Pareto
frontier and runs in polynomial time in the size of the input.

The Pareto frontier of an instance of a multiobjective optimization problem
may contain an arbitrarily large number of solutions. On the contrary, for
every $\varepsilon>0$ there exists a $(1+\varepsilon)$-approximation of the
Pareto frontier that consists of a number of solutions that is polynomial in
the size of the instance and in $\frac{1}{\varepsilon}$ (under reasonable
assumptions). An explicit proof for this observation was given by
\ Papadimitriou and Yannakakis~\cite{py00}. Consequently, a PTAS for a
multiobjective optimization problem does not only have the advantage of
computing a provably good approximation in polynomial time, but also has a
good chance of presenting a reasonably small set of solutions.

Our PTAS for the Pareto frontier of problem $\Pi_{2}$ has even the stronger
property that for every sequence $\sigma^{\prime}(I)\in\mathcal{P}$, the set
$\mathcal{F_{\varepsilon}}$~contains a sequence such that (\ref{pareto1})
holds and (\ref{pareto2}) is replaced by the inequality
\begin{equation}
\label{pareto3}C_{\max}(\sigma(I)) \leq C_{\max} (\sigma^{\prime}(I)).
\end{equation}
\bigskip

\noindent\textbf{Algorithm $PTAS2(I,\varepsilon)$}

\begin{description}
\item[\textsc{{Input:}}] An instance $I$ of $n$ jobs and an accuracy
$\varepsilon$.

\item[\textsc{{Output:}}] A PTAS for the Pareto frontier of problem $\Pi_{2}$.
\end{description}

\begin{enumerate}
\item Run Schrage's algorithm for all instances in $\mathcal{I}$ and store the
solutions in set $\mathcal{F_{\varepsilon}}$.

\item Remove all sequences from $\mathcal{F_{\varepsilon}}$ which are
dominated by other sequences in $\mathcal{F_{\varepsilon}}$.
\end{enumerate}

\begin{theorem}
\label{theo2} Algorithm PTAS2 yields a PTAS for problem $\Pi_{2}$ such that
for every sequence $\sigma^{\prime}(I)\in\mathcal{P}$, the set
$\mathcal{F_{\varepsilon}}$~contains a sequence $\sigma(I)$ such that
(\ref{pareto1}) and (\ref{pareto3}) hold. PTAS2 has polynomial running time
for fixed $\varepsilon$.
\end{theorem}

\begin{proof}
The running time of PTAS2 is polynomial because the set $\mathcal{I}$ contains
only a polynomial number of instances. Let $\sigma^{\prime}(I)$ be a Pareto
optimal sequence for problem $\Pi_{2}$ with $C_{\max}(\sigma^{\prime}(I))=d$.
Algorithm $PTAS1(I,d,\varepsilon)$ of Section~3 outputs a sequence $\sigma(I)$
which fulfills both (\ref{pareto1}) and (\ref{pareto3}). Since sequence
$\sigma(I)$ is also found in Step~1 of Algorithm PTAS2, the theorem follows.
\end{proof}

\section{PTAS for the third scenario: $1,h_{1}|r_{j}|L_{\max}$}

The third scenario $1,h_{1}|r_{j}|L_{\max}$, denoted as $\Pi_{3}$, is studied
in this section. The proposed PTAS for this scenario is related to the first
one and it is based on several steps. In the remainder of this section, $I$
denotes a given instance of $\Pi_{3}$ and $\varepsilon$ is the desired
accuracy. Since a 2-approximation for $1,h_{1}|r_{j}|L_{\max}$ has been
established in~\cite{KH}, we consider only $\varepsilon<1$. Without loss of
generality, in any instance $I$ of problem $\Pi_{3}$, we assume that
$r_{j}\notin\lbrack T_{1},T_{2}[$ for $j\in J$ and that if $r_{j}<T_{1}$, then
the inequality $r_{j}+p_{j}\leq T_{1}$ should hold. Otherwise, in both cases,
$r_{j}$ would be set equal to $T_{2}$. Thus, jobs in $J$ can be partitioned in
two disjoint subsets $X$ and $Y$ where $X=\{j\in J|r_{j}+p_{j}\leq T_{1}\}$
and $Y=\{j\in J|r_{j}\geq T_{2}\}$.

Finally, we assume that $1/\varepsilon$ is integer ($1/\varepsilon=f$) and
that every instance $I$ of problem $\Pi_{3}$ respects the conditions expressed
in the following proposition.\ \ 

\begin{proposition}
With no ($1+\varepsilon$)-loss, every instance $I$ of $\Pi_{3}$ contains at
most $f$ different tails from the set $\{\varepsilon\overline{q}%
,2\varepsilon\overline{q},3\varepsilon\overline{q},...,\overline{q}\}$, where
$\overline{q}=\max_{j\in J}\{q_{j}\}$.
\end{proposition}

\begin{proof}
We simplify the instance $I$ as follows. Split the interval $[0,\max_{j\in
J}\{q_{j}\}]$ into $1/\varepsilon$ equal length intervals and round up every
tail $q_{j}$ to the next multiple of $\varepsilon\overline{q}$. Clearly, every
tail will not be increased by more than $\varepsilon\overline{q}%
\leq\varepsilon L_{\max}^{\ast}(I)$. Then, we obtain that the modified
instance has an optimal solution of a maximum lateness less than or equal to
$\left(  1+\varepsilon\right)  L_{\max}^{\ast}(I)$.
\end{proof}

Since all the jobs of $Y$ should be scheduled after $T_{2}$ in any feasible
solution of $I$, our PTAS will focus on subset $X$. More precisely, this PTAS
will be based on guessing the disjoint subsets $X_{1}$ and $X_{2}$
($X_{1}\subset X$, $X_{2}\subset X$ and $X_{1}\cup X_{2}=X$) such that $X_{1}$
(respectively $X_{2}$) contains the jobs to be scheduled before $T_{1}$
(respectively after $T_{2}$). In the optimal solution, we will denote the jobs
of $X$ to be scheduled before $T_{1}$ (respectively after $T_{2}$) by
$X_{1}^{\ast}$ (respectively by $X_{2}^{\ast}$). It is clear that if we are
able to guess correctly $X_{1}^{\ast}$ in polynomial time, then it is possible
to construct a PTAS for $\Pi_{3}$. Indeed, we can associate the scheduling of
the jobs in $X_{1}^{\ast}$ before $T_{1}$ with a special instance $I_{1}%
^{\ast}$ of $\Pi_{1}$ where the jobs to be scheduled are those in subset
$X_{1}^{\ast}$ and the deadline is equal to $T_{1}$. On the other hand, the
scheduling of the other jobs in $X_{2}^{\ast}\cup Y$ after $T_{2}$ can be seen
as a special instance $I_{2}^{\ast}$ of problem $\Pi_{0}$, with all release
times greater or equal to $T_{2}$. Consequently,%

\begin{equation}
L_{\max}^{\ast}(I)=\max\{L_{\max}^{\ast}(I_{1}^{\ast}),L_{\max}^{\ast}%
(I_{2}^{\ast})\}
\end{equation}

Clearly, the optimal sequence $\sigma_{1}^{\ast}$ of $I_{1}^{\ast}$ should
satisfy the feasibility condition
\begin{equation}
C_{\max}(\sigma_{1}^{\ast}(I_{1}^{\ast}))\leq T_{1}.
\end{equation}

As a consequence, by applying PTAS1 with $d=T_{1}$ or another existing PTAS
for $\Pi_{0}$ (for example PTAS1 with $d$ set to $\infty$ or the PTAS by Hall
and Shmoys~\cite{Hall}), we would get a PTAS for problem $\Pi_{3}$.
Unfortunately, it is not possible to guess the exact subset $X_{1}^{\ast}$ in
a polynomial time, since the potential number of the candidate subsets $X_{1}$
is exponential. Nevertheless, we will show later that we can guess in a
polynomial time another close/approximate subset $X_{1}^{\#}$ for which the
associated instance $I_{1}^{\#}$ of $\Pi_{1}$ and its optimal solution
$\sigma_{1}^{\#}$ verify the following relations:%

\begin{equation}
L_{\max}(\sigma_{1}^{\#}(I_{1}^{\#}))\leq\left(  1+\varepsilon\right)
L_{\max}(\sigma_{1}^{\ast}(I_{1}^{\ast})),
\end{equation}

\begin{equation}
C_{\max}(\sigma_{1}^{\#}(I_{1}^{\#}))\leq C_{\max}(\sigma_{1}^{\ast}%
(I_{1}^{\ast})).
\end{equation}

As the core of our problem of designing a PTAS for $\Pi_{3}$ is to schedule
correctly the jobs of $X$, our guessing approach will be based on a specific
structure on $X$. In the next paragraph, we will describe such a structure.
Then, we present our PTAS and its proof.

\subsection{Defining a structure for the jobs in $X$}

In this subsection, we will focus on the jobs of $X$ and partition them into
subsets. Recall that $P=\sum_{j=1}^{n}p_{j}$. Set $\delta=\frac{\varepsilon
^{2}P}{4}$. A job $j\in X$ is called a \emph{big} job if $p_{j}>\delta$. Let
$B\subset X$ denote the set of big jobs. Obviously, $|B|<4/\varepsilon^{2}$.
The set $S=X\backslash B$ represents the subset of \emph{small} jobs. We
partition $S$ into $f$ subsets $S(k)$, $1\leq k\leq f$, where jobs in $S(k)$
have identical tails of length $k\varepsilon\overline{q}$.

In the next step, jobs in $S(k)$ are sorted in non-decreasing order of release
times. In other words, jobs in $S(k)$ are encoded by pairs of integers
$(k,1),(k,2),\ldots,(k,|S(k)|)$ such that
\begin{equation}
r_{(k,1)}\leq r_{(k,2)}\leq\ldots\leq r_{(k,|S(k)|)}. \label{release}%
\end{equation}
Set $m(k)=\lceil p(S(k))/\delta\rceil$. Hence, $m(k)$ is in $O(4/\varepsilon
^{2})$. For every $k$ and $z$ with $1\leq k\leq f$ and $1\leq z\leq m(k)$, the
subset $S_{k,z}=\{(k,1),(k,2),...,(k,e(z))\}$ is composed of the $e(z)$ jobs
of $S(k)$ with the smallest release times, such that its total processing time
$p\left(  S_{k,z}\right)  =\sum_{l=1}^{e(z)}p_{(k,l)}$ fulfills
\begin{equation}
(z-1)\delta<p\left(  S_{k,z}\right)  \leq z\delta, \label{ez}%
\end{equation}
and $e(z)$ is the largest integer for which (\ref{ez}) holds. Moreover, define
$S_{k,0}=\emptyset$ and $e(0)=0$. Note that $S_{k,m(k)}=S(k)$.

\subsection{Description of the PTAS for $\Pi_{3}$}

As we mentioned before, our proposed PTAS for problem $\Pi_{3}$ is based on
guessing the subset of jobs $X_{1}\subset X$ to be scheduled before $T_{1}$
($X_{2}=X\backslash X_{1}$ will be scheduled after $T_{2}$). For every guessed
subset $X_{1}\subset X$ in our PTAS, we associate an instance $I_{1}$ of
$\Pi_{1}$ with $d=T_{1}$. To the complementary subset $X_{2}\cup Y$ we
associate an instance $I_{2}$ of $\Pi_{0}$ after setting all the release times
in $X_{2}$ to $T_{2}$. In order to solve the generated instances $I_{2}$, we
use PTAS1 with $d=\infty$. (Alternatively, we could also use the existing PTAS
by Hall and Shmoys for $\Pi_{0}$. More details on this PTAS are available in
\cite{Hall}.) We will abbreviate $PTAS1(I,\infty,\varepsilon)$ by
$PTAS0(I,\varepsilon)$. For instances $I_{1}$ we apply $PTAS1(I_{1}%
,T_{1},3\varepsilon/5)$ and for instances $I_{2}$, we apply $PTAS0(I_{2}%
,\varepsilon/3)$. The guessing of $X_{1}$ and the details of our PTAS are
described in procedure PTAS3.\bigskip\ 

\noindent\textbf{Algorithm $PTAS3(I,T_{1},T_{2},\varepsilon)$}

\begin{description}
\item[\textsc{{Input:}}] An instance $I$ of $n$ jobs, integers $T_{1}, T_{2}$
with $T_{1}\leq T_{2}$ and an accuracy $\varepsilon$.

\item[\textsc{{Output:}}] A sequence $\sigma(I)$ with $L_{\max}(\sigma
(I))\leq(1+\varepsilon)L_{\max}^{\ast}\left(  I\right)  $ respecting the
non-availability interval $]T_{1},T_{2}[$.
\end{description}

\begin{enumerate}
\item The big jobs $B_{1}$ to be included in $X_{1}$ are guessed from $B$ by
complete enumeration. For every guess of $B_{1}$, the jobs $B\backslash B_{1}$
will be a part of $X_{2}$.

\item The small jobs to be included in $X_{1}$ are guessed for each subset
$S(k)$, $1\leq k\leq f$, by choosing one of the $m(k)+1$ possible subsets
$S_{k,z}$, $0\leq z\leq m(k)$. Thus, $X_{1}=B_{1}\cup_{k=1}^{f}S_{k,z}$ and
$X_{2}=\left(  B\backslash B_{1}\right)  \cup_{k=1}^{f}\left(  S(k)\backslash
S_{k,z}\right)  $. Set all the release times of jobs in $X_{2}$ to $T_{2}$.

\item With every guessed subset $X_{1}$ in Step~2, we associate an instance
$I_{1}$ of $\Pi_{1}$ with $d=T_{1}$. With the complementary subset $X_{2}\cup
Y$ we associate an instance $I_{2}$ of $\Pi_{0}$. Apply $PTAS1(I_{1}
,T_{1},3\varepsilon/5)$ for solving instance $I_{1}$ of $\Pi_{1}$ and
$PTAS0(I_{2},\varepsilon/3)$ for solving instance $I_{2}$ of $\Pi_{0}$. Let
$\sigma_{1}$\ and $\sigma_{2}$ be the obtained schedules from $PTAS1(I_{1}
,T_{1},3\varepsilon/5)$ and $PTAS0(I_{2},\varepsilon/3)$, respectively. If
they are feasible, then merge $\sigma_{1}$ and $\sigma_{2}$ to get $\sigma$,
which represents a feasible solution for instance $I$ of $\Pi_{3}$.

\item Return $\sigma(I)$, the best feasible schedule from all the complete
merged schedules for $I$ in Step~3.
\end{enumerate}

\subsection{Proof of the PTAS for $\Pi_{3}$}

Denote the jobs from $S(k)$ which are scheduled before $T_{1}$ in the optimal
solution by $S_{k,\ast}$. If $0\leq p\left(  S_{k,\ast}\right)  <p(S_{k}%
)=p\left(  S_{k,m(k)}\right)  ,$then there is an integer $z_{k}\in
\{0,1,\ldots,m(k)-1\}$ such that
\begin{equation}
z_{k}\delta\leq p\left(  S_{k,\ast}\right)  <(z_{k}+1)\delta. \label{ez*}%
\end{equation}

We will see in the following lemma that we can get a sufficient approximation
by taking the jobs of $S_{k,z_{k}}$, $z_{k}=0,1,\ldots, m(k)$, instead of
those in $S_{k,\ast}$. In particular, we will see that it is possible to
schedule heuristically before $T_{1}$ the jobs in $S_{k,z_{k}}$ instead of
those in $S_{k,\ast}$ without altering the optimal performance too much. By
selecting jobs with minimal release times we maintain the feasibility of
inserting these jobs before $T_{1}$ instead of $S_{k,\ast}$. Inequalities
(\ref{ez}) and (\ref{ez*}) imply that there is an integer $z_{k}%
\in\{0,1,\ldots,m(k)-1\}$ such that%

\begin{equation}
p\left(  S_{k,z_{k}}\right)  \leq p\left(  S_{k,\ast}\right)  \leq p\left(
S_{k,z_{k}}\right)  +2\delta\label{ecart}%
\end{equation}
holds for $0\leq p\left(  S_{k,\ast}\right)  <p\left(  S_{k}\right)  $.

Note that if $S_{k,\ast}=S(k)$ holds, the guessing procedure ensures that for
$z_{k}=m(k)$ we get $S_{k,m(k)}=S(k)$.

Recall that $X_{1}^{\ast}$ (respectively $X_{2}^{\ast}\cup Y$) is the subset
of jobs to be scheduled before $T_{1}$ (respectively after $T_{2}$) in the
optimal schedule $\sigma^{\ast}(I)$ \ of a given instance $I$ of $\Pi_{3}$.
Let $I_{1}^{\ast}$, $I_{2}^{\ast}$ denote the corresponding instances of
problems $\Pi_{1}$ and $\Pi_{0}$, and $\sigma_{1}^{\ast}(I_{1}^{\ast})$,
$\sigma_{2}^{\ast}(I_{2}^{\ast})$ the associated optimal schedules,
respectively. Then, we get.

\begin{lemma}%
\index{lemma|see{lemma}}%
There is a subset $X_{1}^{\#}\subset X$, generated in Step~2 of PTAS3, for
which we can construct two feasible schedules $\sigma_{1}^{\#}(I_{1}^{\#})$
and $\sigma_{2}^{\#}(I_{2}^{\#})$ such that:
\begin{equation}
L_{\max}(\sigma_{1}^{\#}(I_{1}^{\#}))\leq L_{\max}(\sigma_{1}^{\ast}%
(I_{1}^{\ast}))+f\delta, \label{xx}%
\end{equation}

\begin{equation}
C_{\max}(\sigma_{1}^{\#}(I_{1}^{\#}))\leq C_{\max}(\sigma_{1}^{\ast}%
(I_{1}^{\ast})), \label{zz}%
\end{equation}

\begin{equation}
L_{\max}(\sigma_{2}^{\#}(I_{2}^{\#}))\leq L_{\max}(\sigma_{2}^{\ast}%
(I_{2}^{\ast}))+2f\delta, \label{yy}%
\end{equation}

where $I_{1}^{\#}$ is an instance of $\Pi_{1}$ where we have to schedule the
jobs of $X_{1}^{\#}$ before $T_{1}$ and $I_{2}^{\#}$ is an instance of
$\Pi_{0}$ where we have to schedule the jobs of $(X\backslash X_{1}^{\#})\cup
Y$ after $T_{2}$.
\end{lemma}

\begin{proof}
See Appendix 1.
\end{proof}

\begin{theorem}
\label{theononavailability} Algorithm PTAS3 is a PTAS for problem $\Pi_{3}$.
\end{theorem}

\begin{proof}
See Appendix 2.
\end{proof}

\section{PTAS for the fourth scenario: $1,ONA|r_{j}|L_{\max}$}

The studied scenario $\Pi_{4}$ can be formulated as follows. An operator has
to schedule a set $J$ of $n$ jobs on a single machine, where every job $j$ has
a processing time $p_{j}$, a release date $r_{j}$ and a tail $q_{j}$. The
machine can process at most one job at a time if the operator must be
available at the starting time and the completion time of such a job. The
operator is unavailable during $]T_{1},T_{2}[$\textbf{.} Preemption of jobs is
not allowed (jobs have to be performed under the non-resumable scenario). As
in the previous scenarios, we aim to minimize the maximum lateness.

If for every $j\in J$ we have $p_{j}<T_{2}-T_{1}$ or $r_{j}\geqslant T_{2}$,
it can be easily checked that $\Pi_{4}$ is equivalent to an instance of
$\Pi_{3}$ for which we can apply the PTAS described in the previous section.
Hence, in the remainder we assume that there are some jobs which have
processing times greater than $T_{2}-T_{1}$ and can complete before $T_{1}$.
Let $\mathcal{K}$ denote the subset of these jobs. Similarly to the
transformation proposed in \cite{Seifa}, the following theorem on the
existence of the PTAS for Scenario $\Pi_{4}$ can be proved.

\begin{theorem}
Scenario $\Pi_{4}$ admits a PTAS.
\end{theorem}

\begin{proof}
We distinguish two subcases:

\begin{itemize}
\item Subcase 1: there exists a job $s\in\mathcal{K}$, called
\textit{straddling job}, such that in the optimal solution it starts not later
than $T_{1}$ and completes not before $T_{2}$.

\item Subcase 2: there is no straddling job in the optimal solution. \ 
\end{itemize}

It is obvious that Subcase 2 is equivalent to an instance of $\Pi_{3}$ for
which we have established the existence of a PTAS. Thus, the last step
necessary to prove the existence of a PTAS for an instance of $\Pi_{4}$ is to
construct a special scheme for Subcase 1. Without loss of generality, we
assume that Subcase~1 holds and that the straddling job $s$ is known (indeed,
it can be guessed in $O(n)$ time among the jobs of $\mathcal{K}$). Moreover,
the time-window of the starting time $t_{s}^{\ast}$ of this job $s$ in the
optimal solution fulfills
\[
t_{s}^{\ast}\in\lbrack T_{2}-p_{s},T_{1}].\label{eqTs}%
\]


If $T_{2}-p_{s}= T_{1}$, i.e., the interval $[T_{2}-p_{s},T_{1}]$ consists of
a single point $T_{1}$, we can apply $PTAS3(I,T_{1},T_{1},\varepsilon)$.
Otherwise, set
\[
t_{s}^{h}=T_{2}-p_{s}+h\frac{T_{1}-T_{2}+p_{s}}{\left\lceil 1/\varepsilon
\right\rceil }\quad h=0,1,\ldots, \left\lceil 1/\varepsilon\right\rceil .
\]
We consider a set of $\left\lceil 1/\varepsilon\right\rceil +1$ instances
$\{\mathcal{I}_{0},\mathcal{I}_{1},\ldots,\mathcal{I}_{\left\lceil
1/\varepsilon\right\rceil }\}$ of $\Pi_{4}$ where in $\mathcal{I}_{h}$ the
straddling job $s$ starts at time $t_{s}^{h}$. Each such instance is
equivalent to an instance of $\Pi_{3}$ with a set of jobs $J-\{s\}$ and a
machine non-availability interval $\Delta_{h}$ with
\[
\Delta_{h}=\left]  T_{2}-p_{s}+h\frac{T_{1}-T_{2}+p_{s}}{\left\lceil
1/\varepsilon\right\rceil },T_{2}+h\frac{T_{1}-T_{2}+p_{s}}{\left\lceil
1/\varepsilon\right\rceil }\right[  .
\]
For every instance from $\{\mathcal{I}_{0},\mathcal{I}_{1},\ldots
,\mathcal{I}_{\left\lceil 1/\varepsilon\right\rceil } \}$, we apply PTAS3 and
select the best solution among all the $\left\lceil 1/\varepsilon\right\rceil
+1$ instances.

If $t_{s}^{\ast} = T_{2}-p_{s}$, PTAS3 applied to $\mathcal{I}_{0}$ is the
right choice. In all other cases, there is an $h\in\{0,1,\ldots,
\lceil1/\varepsilon\rceil\}$ with $t_{s}^{\ast}\in]t_{s}^{h},t_{s}^{h+1}]$.
Delaying $s$ and the next jobs in the optimal schedule of $\mathcal{I}_{h+1}$,
$h=0,1,\ldots,\left\lceil 1/\varepsilon\right\rceil -1$, will not cost more
than
\[
\frac{T_{1}-T_{2}+p_{s}}{\left\lceil 1/\varepsilon\right\rceil }
\leq\varepsilon\left(  T_{1}-T_{2}+p_{s}\right)  \leq\varepsilon p_{s}.
\]
Thus, the solution $\Omega_{h}$ obtained by PTAS3 for $\mathcal{I}_{h}$,
$h=1,2,\ldots,\left\lceil 1/\varepsilon\right\rceil $) is sufficiently close
to an optimal schedule for Subcase 1 if $s$ is the straddling job and
$t_{s}^{\ast}\in] t_{s}^{h},t_{s}^{h+1}]$. As a conclusion, Subcase~1 has a PTAS.
\end{proof}

\section{Conclusions}

In this paper, we considered important single-machine scheduling problems
under release times and tails assumptions, with the aim of minimizing the
maximum lateness. Four variants have been studied under different scenarios.
In the first scenario, all the jobs are constrained by a common deadline. The
second scenario consists in finding the Pareto front where the considered
criteria are the makespan and the maximum lateness. In the third scenario, a
non-availability interval constraint is considered (breakdown period, fixed
job or a planned maintenance duration). In the fourth scenario, the
non-availibility interval is related to the operator who is organizing the
execution of jobs on the machine, which implies that no job can start, and
neither can complete during the operator non-availability period. For each of
these four open problems, we establish the existence of a PTAS.

As a perspective, the study of min-sum scheduling problems in a similar
context seems to be a challenging subject.



\bigskip

\textbf{APPENDEX 1: Proof of Lemma 4}

\begin{proof}
Let $B_{1}^{\ast}$ be the subset of big jobs contained in $X_{1}^{\ast}$. The
other jobs of $X_{1}^{\ast}$ are small and belong to $\cup_{k=1}^{f}S(k)$.
Hence,
\[
X_{1}^{\ast}=B_{1}^{\ast}\cup\bigcup\limits_{k=1}^{f}S_{k,\ast}%
\]
and $X_{1}^{\#}$ will have the following structure:
\[
X_{1}^{\#}=B_{1}^{\ast}\cup\bigcup\limits_{k=1}^{f}S_{k,z_{k}}.\label{guess}%
\]
Obviously, subset $B_{1}^{\ast}$ can be guessed correctly in Step~2 of PTAS3
since PTAS3 considers all possible subsets of $B$. Therefore, for proving that
subset $X_{1}^{\#}$ is a good guess, we need to prove that the guessing of
each subset $S_{k,z_{k}}$ is close enough to $S_{k,\ast}$, but with smaller
total processing time.

First, we observe that if $S_{k,\ast}=\varnothing$ or $S_{k,\ast}=S(k)$, then
the guessing is perfect, i.e. $S_{k,\ast}=S_{k,z}$.

If $S_{k,\ast}\not =\varnothing$, we can determine intervals where only jobs
of $S_{k,\ast}$ are scheduled in the optimal schedule $\sigma_{1}^{\ast}%
(I_{1}^{\ast})$ and which contain no idle time. Denote those intervals by
\[
G_{1,k} = [a_{1,k}^{\ast},b_{1,k}^{\ast}], G_{2,k} = [a_{2,k}^{\ast}%
,b_{2,k}^{\ast}], \ldots, G_{\gamma_{k},k} = [a_{\gamma_{k},k}^{\ast
},b_{\gamma_{k},k }^{\ast}],\quad k=1,\ldots,f,
\]
where $\gamma_{k}$ is a certain finite integer. Set
\[
G_{k} =\bigcup_{u=1}^{\gamma_{k}}G_{u,k}, \quad k=1,\ldots,f.
\]
Otherwise, in case $S_{k,\ast}=\varnothing$, we set $G_{k}=\emptyset$.
Finally, define
\[
G =\bigcup_{k=1}^{f}G_{k}.
\]

In addition, we recall that $S_{k,z_{k}}=\{(k,1),(k,2),...,(k,e(z_{k}))\}$. In
our feasible schedule $\sigma_{1}^{\#}(I_{1}^{\#})$, we will schedule the jobs
of this subset $S_{k,z_{k}}$ according to the following heuristic procedure
\emph{Schedule}$^{\mathbf{\#}}$. It assigns the jobs to the intervals
$G_{u,k}$ but throughout the procedure starting and finish times of these
intervals may be changed (e.g. when there is some overload of jobs) or they
are simply shifted to the right. \bigskip

\textbf{Procedure }\emph{Schedule}$^{\mathbf{\#}}$

\begin{enumerate}
\item \textbf{For} $k=1$ to $f$ do:

\item $\alpha_{0}:=0;$

\item \textbf{For} $u=1$ to $u=\gamma_{k}$ do:

\begin{enumerate}
\item Process remaining jobs from $S_{k,z_{k}}$ (in non-decreasing order of
release dates) starting at time $a_{u,k}^{\ast}+\alpha_{u-1}$ until the first
job finishes at time $t\geq b_{u,k}^{\ast}$ or all jobs from $S_{k,z_{k}}$ are
assigned. Let job $(k,l)$ be the last scheduled job from this set with
completion time $C_{(k,l)}$.

\item \textbf{If} $l=e(z_{k})$ set $G_{u,k}=[a_{u,k}^{\ast}+\alpha
_{u-1},\,b_{u,k}^{\ast}]$ and \textbf{Stop}: All items from $S_{k,z_{k}}$ are assigned.

\item \textbf{Else} set $\alpha_{u}=C_{(k,l)}-b_{u,k}^{\ast}$ and
$G_{u,k}=[a_{u,k}^{\ast}+\alpha_{u-1},\,b_{u,k}^{\ast}+\alpha_{u}%
]=[a_{u,k}^{\ast}+\alpha_{u-1},\,C_{(k,l)}]$.

\item Shift intervals from $G$ which are located between $b_{u,k}^{\ast}$ and
$a_{u+1,k}^{\ast}$ without any modification of their lengths by $\alpha_{u+1}$
to the right.

\item End \textbf{If/Else}
\end{enumerate}

\item End \textbf{For}


\item End \textbf{For}
\end{enumerate}

\bigskip


Notice that the value $\alpha_{u}$ represents the ``overload'' of the $u$-th
interval $G_{u,k}$ in a given set $G_{k}$. The following interval $G_{u+1,k}$
is then shortened by the same amount $\alpha_{u}$. Consequently, only the
intervals between $G_{u,k}$ and $G_{u+1,k}$ are shifted to the right, but no
others. Moreover, the total processing time assigned until $b_{u+1,k}^{\ast}$
in the $k$-th iteration of \emph{Schedule}$^{\#}$ is not greater than the
original total length of the intervals $G_{1,k}, G_{2,k},\ldots, G_{u+1,k}$.
Since for every family $k$ the jobs in $S_{k,z_{k}}$ have the smallest release
times and since by (\ref{release}) these jobs are sorted in non-decreasing
order of release times, no idle time is created during Step~3. Together with
the first inequality of (\ref{ecart}) this guarantees that the all jobs of
family $k$ are processed without changing the right endpoint of the last
interval $G_{\gamma_{k},k}$, i.e., creating no overload in this interval. That
ensures the feasibility of the schedule yielded by \emph{Schedule}%
$^{\mathbf{\#}}$ and that $G_{\gamma_{k},k}=[a_{\gamma_{k},k}^{\ast}%
+\alpha_{\gamma_{k}-1},\,b_{\gamma_{k},k}^{\ast}]$.

In each iteration of the procedure the jobs will be delayed by at most
$\max_{u\leq\gamma_{k}}\{\alpha_{u}\}\leq\delta$. Since this possible delay
may occur for every family $k$, $k=1,2,...,f$, the possible overall delay for
every job, after the $f$ applications of procedure \emph{Schedule}$^{\#}$ is
therefore bounded by $f\delta$. Hence, (\ref{xx}) follows.

The finishing time of the last interval among the intervals $G_{\gamma_{k},k},
k=1,\ldots, f$, generated by \emph{Schedule}$^{\#}$ determines the makespan of
instance $I_{1}^{\#}$. But this interval is never shifted to the right and we
obtain (\ref{zz}).

Recall that if $S_{k,\ast}=S(k)$, we get with $z_{k}=m(k)$ that $S_{k,m(k)}%
=S(k)$. Hence, in the case that no jobs of family $k$ are scheduled after
$T_{2}$ in $\sigma_{2}^{\ast}(I_{2}^{\ast})$, also $\sigma_{2}^{\#}(I_{2}%
^{\#})$ contains no jobs of family $f$.

Hence, assume in the following that $S(k)\backslash S_{k,\ast}\not =\emptyset$
and let $[c_{1,k}^{\ast},d_{1,k}^{\ast}],[c_{2,k}^{\ast},d_{2,k}^{\ast
}],\ldots,[c_{\lambda_{k},k}^{\ast},d_{\lambda_{k},k}^{\ast}]$ be the
intervals in which the jobs of $S(k)\backslash S_{k,\ast}$ are scheduled in
the optimal schedule $\sigma_{2}^{\ast}(I_{2}^{\ast})$. The schedule
$\sigma_{2}^{\#}(I_{2}^{\#})$ will be created similarly to procedure
\emph{Schedule}$^{\#}$.

As a consequence of performing possibly less amount of jobs from every $S(k)$
before $T_{1}$ in the feasible schedule $\sigma_{1}^{\#}(I_{1}^{\#})$, an
additional small quantity of processing time must be scheduled after $T_{2}$
compared to the optimal solution. This amount is bounded by $2\delta$ by the
second inequality of (\ref{ecart}). Therefore, in each iteration $k$ we start
by enlarging the last interval by $2\delta$, i.e., $[c_{\lambda_{k},k}^{\ast
},d_{\lambda_{k},k}^{\ast}+2\delta]$ instead of $[c_{\lambda_{k},k}^{\ast
},d_{\lambda_{k},k}^{\ast}]$, and shifting those jobs starting after
$d_{\lambda_{k},k}^{\ast}$ in $\sigma_{2}^{\ast}(I_{2}^{\ast})$ by the same
distance $2\delta$ to the right. Then, all what we have to do in order to
continue the construction of our feasible schedule $\sigma_{2}^{\#}(I_{2}%
^{\#})$ is to insert the remaining jobs from every $S(k)\backslash S_{k,z_{k}%
}$ in the corresponding intervals $[c_{1,k}^{\ast},d_{1,k}^{\ast}%
],[c_{2,k}^{\ast},d_{2,k}^{\ast}],\ldots,[c_{\lambda_{k},k}^{\ast}%
,d_{\lambda_{k},k}^{\ast}+2\delta]$ like in iteration $k$ of procedure
\emph{Schedule}$^{\#}$.

Analogously to \emph{Schedule}$^{\#}$ intervals between $d_{u,k}^{\ast}$ and
$c_{u+1,k}^{\ast}$ are shifted to the right by some distance bounded by
$\delta$. Thus, in a certain iteration $k$ jobs are shifted to the right by
not more than $\max\{2\delta,\delta\}=2\delta$. Since the possible delay may
occur for every job family $k$, the possible overall delay for every job
scheduled in $\sigma_{2}^{\#}(I_{2}^{\#})$ is bounded by $2f\delta$. This
shows that inequality (\ref{yy}) is valid.
\end{proof}

\bigskip

\textbf{APPENDIX 2: Proof of Theorem 5}

\begin{proof}
The big jobs $B_{1}$ to be included in $X_{1}$ are completely guessed from $B$
in $O(2^{4/\varepsilon^{2}})$. The small jobs to be included in $X_{1}$ are
guessed from every subset $S(k)$ by considering only the $m(k)+1$ possible
subsets $S_{k,z}$. This can be done for every guessed $B_{1}$ in $O(\left(
m(k)+1\right)  ^{f})=O(\left(  4/\varepsilon^{2}\right)  ^{1/\varepsilon})$
time. In overall, the guessing of all the subsets $X_{1}$ will be done in
$O(2^{4/\varepsilon^{2}}.\left(  4/\varepsilon^{2}\right)  ^{1/\varepsilon})$.
Hence, from the computational analysis of the different steps of PTAS3, it is
clear that this algorithm runs polynomially in the size of instance $I$ of
$\Pi_{3}$ for a fixed accuracy $\varepsilon>0$.

Let us now consider the result of Lemma $4$. We established that a subset
$X_{1}^{\#}\subset X$ is generated in Step~2 of PTAS3 and that two feasible
schedules $\sigma_{1}^{\#}(I_{1}^{\#})$ and $\sigma_{2}^{\#}(I_{2}^{\#})$
exist such that (\ref{xx}), (\ref{zz}) and (\ref{yy}) hold. By
Theorem~\ref{theo1} we have%

\[
L_{\max}\left(  PTAS1(I_{1}^{\#},T_{1},3\varepsilon/5)\right)  \leq\left(
1+3\varepsilon/5\right)  L_{\max}(\sigma_{1}^{\#}(I_{1}^{\#})).
\]

From (\ref{xx}), we then deduce that
\begin{equation}
L_{\max}\left(  PTAS1(I_{1}^{\#},T_{1},3\varepsilon/5)\right)  \leq\left(
1+3\varepsilon/5\right)  \left(  L_{\max}(\sigma_{1}^{\ast}(I_{1}^{\ast
}))+f\delta\right)  . \label{xx1}%
\end{equation}

Moreover, we have%
\[
L_{\max}\left(  PTAS0(I_{2}^{\#},\varepsilon/3)\right)  \leq\left(
1+\varepsilon/3\right)  L_{\max}(\sigma_{2}^{\#}(I_{2}^{\#})).
\]

From (\ref{yy}) we establish that
\begin{equation}
L_{\max}\left(  PTAS0(I_{2}^{\#},\varepsilon/3)\right)  \leq\left(
1+\varepsilon/3\right)  \left(  L_{\max}(\sigma_{2}^{\ast}(I_{2}^{\ast
}))+2f\delta\right)  . \label{yy1}%
\end{equation}

By using (\ref{zz}) we get
\begin{equation}
C_{\max}(\sigma_{1}^{\#}(I_{1}^{\#}))\leq C_{\max}(\sigma_{1}^{\ast}%
(I_{1}^{\ast}))\leq T_{1}. \label{eps3}%
\end{equation}

Inequality (\ref{eps3}) implies that the algorithm $PTAS1(I_{1}^{\#}%
,T_{1},3\varepsilon/5)$ yields a feasible solution with $C_{\max}%
(PTAS1(I_{1}^{\#},T_{1},3\varepsilon/5))\leq T_{1}$.

Let us now define the global solution $\sigma^{\#}(I)$ for instance $I$ (of
$\Pi_{3}$) as the union of the two feasible subsolutions $PTAS1(I_{1}%
^{\#},T_{1},3\varepsilon/5)$ and $PTAS0(I_{2}^{\#},\varepsilon/3)$. The
following relation holds
\[
L_{\max}\left(  \sigma^{\#}(I)\right)  =\max\{L_{\max}(PTAS1(I_{1}^{\#}%
,T_{1},3\varepsilon/5)),L_{\max}\left(  PTAS0(I_{2}^{\#},\varepsilon
/3)\right)  \}.
\]

Thus, by applying (\ref{xx1}) and (\ref{yy1}), we deduce that
\begin{equation}
L_{\max}\left(  \sigma^{\#}(I)\right)  \leq\max\{\left(  1+3\varepsilon
/5\right)  \left(  L_{\max}(\sigma_{1}^{\ast}(I_{1}^{\ast}))+f\delta\right)
,\left(  1+\varepsilon/3\right)  \left(  L_{\max}(\sigma_{2}^{\ast}%
(I_{2}^{\ast}))+2f\delta\right)  \}. \label{aux}%
\end{equation}
From (\ref{aux}) and by observing that $f\delta\leq\varepsilon L_{\max}^{\ast
}(I)/4$, we obtain for $\varepsilon\leq1$ the final relations%

\begin{align*}
L_{\max}\left(  \sigma^{\#}(I)\right)   &  \leq\max\{L_{\max}(\sigma_{1}%
^{\ast}(I_{1}^{\ast}))+\varepsilon L_{\max}^{\ast}(I),L_{\max}(\sigma
_{2}^{\ast}(I_{2}^{\ast}))+\varepsilon L_{\max}^{\ast}(I)\}\\[1ex]
&  =\max\{L_{\max}(\sigma_{1}^{\ast}(I_{1}^{\ast})),L_{\max}(\sigma_{2}^{\ast
}(I_{2}^{\ast}))\}+\varepsilon L_{\max}^{\ast}(I)=(1+\varepsilon)L_{\max
}^{\ast}(I).
\end{align*}
This finishes the proof of Theorem~\ref{theononavailability}.
\end{proof}

\end{document}